\documentclass[a4paper,12pt,twoside]{article}
\usepackage{amsmath,amssymb,amsthm}
\usepackage{graphics,epsfig,calc}
\usepackage{upref}
\usepackage{showkeys}

\newtheorem{theorem}{Theorem}[section]
\newtheorem{qtheorem}{QTheorem}[section]

\newtheorem{definition}[theorem]{Definition}

\newtheorem{lemma}[theorem]{Lemma}

\newtheorem{remark}[theorem]{Remark}

\newtheorem{cor}[theorem]{Corollary}
\newtheorem{prop}[theorem]{Proposition}

\newcommand{\bt}{\begin{theorem}}
 \newcommand{\et}{\end{theorem}}
\newcommand{\bqt}{\begin{qtheorem}}
 \newcommand{\eqt}{\end{qtheorem}}

\newcommand{\bp}{\begin{pro}}
 \newcommand{\ep}{\end{pro}}

\newcommand{\bl}{\begin{lemma}}
 \newcommand{\el}{\end{lemma}}
\newcommand{\bc}{\begin{cor}}
 \newcommand{\ec}{\end{cor}}

\newcommand{\norm}[2]{\left\Vert{#1}\right\Vert_{#2}}

\newcommand{\R}{\mathbb{R}}
\newcommand{\Z}{\mathbb{Z}}
\newcommand{\C}{\mathbb{C}}

\newcommand{\cB}{{\cal B}}

\newcommand{\al}{\alpha}

\newcommand{\de}{\delta}
\newcommand{\ve}{\varepsilon}

\newcommand{\om}{\omega}

\newcommand{\si}{\sigma}
\newcommand{\Si}{\Sigma}

\newcommand{\ds}{\displaystyle}
\newcommand{\la}{\label}

\newcommand{\beqn}{\begin{eqnarray}}
\newcommand{\eeqn}{\end{eqnarray}}
\newcommand{\be}{\begin{equation}}
\newcommand{\ee}{\end{equation}}
\newcommand{\pa}{\partial}
\newcommand{\fr}{\frac}

\newcommand{\re}{\ref}
\newcommand{\ti}{\tilde}

\renewcommand{\Re}{\mathop{\mathrm{Re}}} 
\renewcommand{\Im}{\mathop{\mathrm{Im}}} 

\DeclareMathOperator{\const}{const}
\DeclareMathOperator{\supp}{supp}

\title{\footnotesize{\textit{}}\hfill\hbox{}
\newline\hbox{}\newline\Large{On Dispersive Estimates for  Discrete\\
Schr\"odinger and Klein-Gordon Equations}}

\begin{document}
\date{}
\maketitle
\begin{center}
{\large E.~A.~Kopylova} \footnote{Supported partly by  FWF, DFG
and RFBR grants}\\
{\it Institute for Information Transmission Problems RAS\\
B.Karetnyi 19, Moscow 101447,GSP-4, Russia}\\
e-mail:~ek@vpti.vladimir.ru
\end{center}

\begin{abstract}
{\em We derive the long-time asymptotics for solutions of the
discrete 3D Schr\"odinger and Klein-Gordon equations.}
\end{abstract}

{\em Keywords}: discrete Schr\"odinger and Klein-Gordon equations,
lattice, Cauchy problem, long-time asymptotics, weighed norms.

{\em 2000 Mathematics Subject Classification}: 39A11, 35L10.

\section{Introduction}

We consider the 3D discrete version of the Schr\"odinger equation,
\begin{equation} \label{Schr}
   \left\{\begin{array}{l}
   i\dot\psi(x,t)=H\psi(x,t):=(-\Delta  + V(x))\,\psi(x,t)
   \\\\
   \psi\bigr|_{t=0} = \psi_0
\end{array}\right|\quad x\in\Z^3,\quad t\in\R,
\end{equation}
where $\Delta$ stands for the difference Laplacian in $\Z^3$, defined by
\begin{equation*}
    \Delta \psi(x)=\sum\limits_{|y-x|=1}\psi(y)-6\psi(x),\quad x\in\Z^3,
    \quad \psi:\Z^3\to\C
\end{equation*}
Denote by ${\cal V}$ the set of real valued functions on the
lattice $\Z^3$ with finite supports. For the potential $V$, we
assume that ${V\in\cal V}$. Then for $\psi_0\in\l^2=\l^2(\Z^3)$ there exists
the unique solution $\psi(x,t)\in C(\R,l^2)$ to Cauchy  problem (\re{Schr}), and
the charge $\Vert\psi(\cdot,t)\Vert_{l^2}=\const$ is conserved.

It is well known that with the help of  the Fourier-Laplace transform in respect to
variable $t$ one can deduce the properties of nonstationary equation from the
properties of resolvent $R(\omega)=( H-\omega)^{-1}$ of the
Schr\"odinger operator $H$.

We are going to use the weighted Hilbert spaces $l^2_{\sigma}=l^2_{\sigma}(\Z^3)$ with the
norms
\begin{equation*}
   \norm{u}{l^2_{\sigma}}=\norm{(1+x^2)^{\sigma/2}u}{l^2},\quad\sigma\in\R.
\end{equation*}
Let us denote $$B(\sigma,\,\sigma')={\cal L}(l^2_{\sigma},\, l^2_{\sigma'})$$
the space of bounded linear operators from $l^2_{\sigma}$ to $l^2_{\sigma'}$.
\smallskip

The  spectrum of the operator $H$ consists of the continuous spectrum
and of the  real eigenvalues  $\mu_j$, $j=1,...,n$. Note that
$n\le N$, where $N$ is the number of points in the support of $V$
(see. \cite {Gl}, Theorem $13^{bis}$, Chapter I).

Note that the continuous spectrum of the operator $H$ coincides
with the interval $[0,12]$, which is the range of
the symbol $\phi(\theta)=4(\sin^2\frac{\theta_1}2+\sin^2\frac{\theta_2}2+\sin^2\frac{\theta_3}2)$
of the difference Laplace operator $H_0=-\Delta$.

We give special attention to the points $\om_k=4k\in [0,12]$, $k=0,1,2,3$,
which are critical values of the symbol, i.e. the values of the symbol in the
critical points.

Our main results are as follows. For ``a generic potential'' $V\in {\cal V}$
(see Definition \ref{gener-def}), we obtain\\
a) the existence of the limits $R(\om\pm i0)$ (``limiting absorption principle'')
on the continuous spectrum  in the norm of $B(\sigma;\,-\sigma)$ with $\sigma>3/2$;\\
b) the Puiseux expansion for the resolvent at the singular
spectral points $\om_k$:
\begin{equation}\label{P}
  R(\om_k+\om)= D_{k}+{\cal O}(\sqrt\om),\quad\om\to 0,
\end{equation}
in the norm of $B(\sigma;\,-\sigma)$ with $\sigma>7/2$.

Then for initial data $\psi_0\in l^2_{\sigma}$ with $\sigma>11/2$
we obtain the following long-time asymptotics:
\begin{equation}\label{full1}
  \Vert e^{-itH}-\sum\limits_{j=1}^n e^{-it\mu_j}P_j\Vert_{{\cal B}(\sigma,-\sigma)}
  ={\cal O}(t^{-3/2}),\quad t\to\infty.
\end{equation}
Here $ P_j$ are the orthogonal projections
in $l^2$ onto the eigenspaces of $H$,
corresponding to the discrete eigenvalues $\mu_j$.

We also obtain similar results for the discrete Klein-Gordon equation:
\begin{equation} \label{KGE}
   \left\{\begin{array}{l}\ddot\psi(x,t)=(\Delta -m^2 - V(x))\,\psi(x,t)
   \\\\
   \psi\bigr|_{t=0} = \psi_0,\;\dot \psi\bigr|_{t=0} = \pi_0
   \end{array}\right|\quad x\in\Z^3,\quad t\in\R.
\end{equation}
Let us comment on previous results in this direction. For the first time the  difference
Schr\"odinger equation was considered by Eskina  \cite{E}. She proved the
limiting absorption principle for matrix elements of the resolvent.
The asymptotic expansion of the matrix element of the resolvent
$R(\omega)$ at the critical points $\om_k$ was obtained by Islami and
Vainberg \cite{IV} in 2D case. They used this expansion to prove the
long time asymptotics for the solutions of the Cauchy problem for
the difference wave equation. The main feature which differs the
present paper from \cite{IV} is that here all asymptotic
expansions hold in the weighted functional spaces $l^2_{\sigma}$, not on compacts.
Such expansions are desirable for the study of nonlinear evolutionary equations.

The asymptotic expansion of the resolvent and the long time
asymptotics \eqref{full1} for hyperbolic PDEs in $\R^n$
(continuous case) were obtained earlier in \cite{LP},
\cite{V}, \cite{V2}, and for the Schr\"odinger equation in
\cite{jeka}, \cite{jene}, \cite{M}; also see
\cite{schlag} for an up-to-date review.
We use the main ideas of the papers.

The results of present paper extend the results of \cite{kkk} and \cite{kkv}
from difference 1D and 2D equations to  difference 3D equations.
In 3D case the analytical problems is more difficult because
of several type of critical points.

The paper is organized as follows. In \S\ref{free-sect} we
prove the limiting absorption principle and derive the Puiseux asymptotic
of the free resolvent. In \S\ref{lap-pert} we extend the results
to perturbed resolvent.
In  \S\ref{lt-as} we prove the long-time asymptotics \eqref{full1}.
In  \S\ref{KG-sect} we consider the discrete Klein-Gordon equation.
In Appendix C we apply the  obtained results to construct  asymptotic scattering states.

\setcounter{equation}{0}
\section{Free resolvent}
\label{free-sect}
We start with an investigation of the unperturbed problem for the
equation  \eqref {Schr} with $V(x)=0$. The discrete Fourier
transform of $u(x)\in l^2(\Z^3)$ is defined by the formula
\begin{equation*}
    \widehat u(\theta)=\sum_{x\in\Z^3}u(x)e^{i\theta x},
    \;\theta\in T^3:=\R^3/2\pi \Z^3,
\end{equation*}
After taking the Fourier transform, the operator $H_0=-\Delta$ becomes the
operator of multiplication by
$\phi(\theta):=6-2\sum\limits_{j=1}^{3}\cos \theta_j
=4\sum\limits_{j=1}^{3}\sin^2\frac{\theta_j}2$:
\begin{equation}\label{FTDelta}
    -\widehat{\Delta u}(\theta)=
     \phi(\theta)\widehat u(\theta),\quad\theta\in T^3.
\end{equation}
Thus,   the spectrum of the operator $H_0$
coincides with the range of the function $\phi$, that is ${\rm Spec}H_0=\Si:=[0,12]$.
Denote by $R_0(\omega)=(H_0-\omega)^{-1}$ the
resolvent of the difference Laplacian. Then the  kernel of the resolvent
$R_0(\omega)$ reads
\begin{equation}\label{R0}
   R_0(\omega,x-y)= \frac 1{8\pi^3}\int\limits_{T^3}
   \frac{e^{-i\theta(x-y)}}{\phi(\theta)-\omega}~d\theta,\;
   \omega\in\C\setminus \Si.
\end{equation}
\begin{lemma}\label{PV}
 The free resolvent $R_0(\omega)$  is an analytic function
 of $\omega\in\C\setminus \Si$
 with the values in $\cB(\si,\si')$ for any $\si,\si'\in\R$.
\end{lemma}
\begin{proof}
  For a fixed $\omega\in\C\setminus \Si$, we have $\phi(\theta)-\om \ne 0$ for
  $\theta\in T^3$. Therefore, $\phi(\theta+i\xi)-\om \ne 0$ for
  $\theta\in T^3$, $\xi\in\R^3$,  if $\xi\not =0$  is sufficiently
small. Hence, the function $1/(\phi(\theta)-\om)$ admits analytic
continuation into a complex neighbourhood of the torus of type
  $\{\theta+i\xi: \theta\in T^3,\;\xi\in\R^3:|\xi|<\de(\om)\}$ with an
  $\de(\om)>0$.  Therefore the Paley-Wiener arguments imply that
  \[
    |R_0(\omega,x-y)|\le C(\delta)e^{-\delta|x-y|}
  \]
  for any $\de<\de(\om)$.
  Hence,  $R_0(\omega)\in \cB(\si,\si')$ by the Schur lemma.
\end{proof}
\subsection{Limiting absorption principle}
\label{lim-ab-pr}
Now we are interested in the traces of the analytic function $R_0(\om)$  at the cut
$\Si$.  Consider
\be\la{RF}
   R_0(\om\pm i\ve,z)=\frac 1{8\pi^3}\int\limits_{T^3}
    \frac{e^{-i\theta z}}{\phi(\theta)-\omega\mp i\ve}~d\theta,
    \quad z\in\Z^3,\quad\om\in\Si,\quad\ve>0.
\ee
Note that the limiting distribution
$\ds\fr 1{\phi(\theta)-\om\mp i 0}$ is well defined if
$\om$  is not a critical value of the function $\phi(\theta)$, i.e.
$\om\not=0,4,8,12$.
The following limiting absorption principle holds:
\begin{prop}\label{LAP}
  For $\sigma>3/2+k$ the following limits exist as $\varepsilon\to 0+$:
  \be\label{lim}
    \pa_\om^kR_0(\omega\pm i\varepsilon)\;\buildrel {\hspace{2mm}\mathcal
    B(\sigma,-\sigma)}\over
    {- \hspace{-2mm} \longrightarrow}\;\pa_\om^k R_0(\omega\pm i0),\quad
    \omega\in \Si\setminus \{0,4,8,12\}.
  \ee
\end{prop}
\begin{proof} i) Let $k=0$.
  First we prove the convergence (\ref{lim}) for any fixed $z$, follow \cite{E}.
  Let $\chi_j(\theta)$, $j=1,...,l$ are the sufficient small partition
  of unity on the torus $T^3$, which will be specified below. Then
  \be\label{esk1}
     R_0(\omega\pm i\ve,z)=\sum\limits_{j=1}^{l}\frac 1{8\pi^3}\int\limits_{D_j}
     \frac{\chi_j(\theta)e^{-i\theta z}}{\phi(\theta)-\omega\mp i\ve}~d\theta=
     \sum\limits_{j=1}^{l}P_{j}(\omega\pm i\ve,z),
  \ee
  where $D_j$ is the support of the function $\chi_j$.
  If   $\{\phi(\theta)=\omega\}\cap D_j=\emptyset$, then
  the function $P_{j}(\omega\pm i\ve,z)$ is continuous for $\ve\geq 0$ and
  \begin{equation}\label{R0j}
      |P_{j}(\omega\pm i\ve,z)|<C_j<\infty,\quad
      z\in Z^3,\quad \ve\geq 0.
  \end{equation}
  Now let  $S_j=\{\phi(\theta)=\omega\}\cap D_j$.
  Then  any $\theta\in D_j$ can be uniquely represented as
  $\theta=s+tn(s)$
  where $s\in S_j$, and $n(s)$ is the external normal vector to
  $S_j$ at the point $s$ of unit length.
  Let us introduce the new variables $(s,t)$. Then
  \be\label{esk2}
    P_{j}(\omega\pm i\ve,z)=\frac 1{8\pi^3}\int\limits_{S_j}e^{-isz}ds
    \int\limits_{-a(s)}^{b(s)}\frac{\chi_j(s+tn(s))e^{-itn(s)z}J(s,t)}
    {t\psi(s+tn(s))\mp i\ve}~dt.
  \ee
  where $J(s,t)$ is the Jacobian,  $\psi$ is the smooth function, and
  $a(s)$, $b(s)>0$.
  Note that $J(s,t)|_{t=0}=1$, and $\psi(s+tn(s))|_{t=0}=|\nabla\phi(s)|\not=0$,
  since $\om\in\Sigma\setminus\{0,4,8,12\}$ is not a critical value of $\phi(\theta)$.
  We will prove the following lemma:
  \begin{lemma}\la{SP}
    Let  $\varphi(t,z)$ be the smooth function satisfies
    \be\la{prp}
      |\varphi(t,z)|\le C,\quad |\partial_t\varphi(t,z)|\le C|z|,
      \quad t\in[-\delta,\delta],\;z\in\Z^3,
    \ee
    and let $\psi(t)$ be the smooth function such that $\psi(t)\not =0$
    if $t\in [-\delta,\delta]$.
    Consider
    $$
    F(\pm\ve,z)=\ds\int\limits_{-\delta}^{\delta}
    \frac{\varphi(t,z)}{t\psi(t)\mp i\ve}dt.
    $$
    Then
    $F(\pm\ve,z)\to F(\pm 0,z)$ as $\varepsilon\to 0+$, $\forall z\in\Z^3$,
    and
    $$
       \sup_{\ve\in (0,1]}|F(\pm\ve,z)|\le C\Big(\ln (1+|z|)+1\Big),\quad z\in\Z^3.
    $$
  \end{lemma}
  \begin{proof}
  Let us rewrite $F(\pm\ve,z)$ as
  \beqn\nonumber
    F(\pm\ve,z)&=&\varphi(0,z)\!\!\int\limits_{-\delta}^{\delta}\!
    \frac{dt}{t\psi(0)\mp i\ve}-\varphi(0,z)\!\!\int\limits_{-\delta}^{\delta}\!
    \frac {(\psi(t)-\psi(0))tdt}{(t\psi(t)\mp i\ve)(t\psi(0)\mp i\ve)}\\
    \nonumber
    &+&\int\limits_{-\delta}^{\delta}\!\frac{\varphi(t,z)-\varphi(0,z)}{t\psi(t)\mp i\ve}dt.
  \eeqn
  Then,
  \be\la{esk3}
\begin{split}
F(\pm\ve,z)\to F(\pm 0,z)&=\pm i\pi\frac{\varphi(0,z)}{\psi(0)}
    -\frac{\varphi(0,z)}{\psi(0)}\!\int\limits_{-\delta}^{\delta}\!
    \frac {\psi(t)\!-\!\psi(0)}{t\psi(t)}dt
\\
&+\int\limits_{-\delta}^{\delta}
    \frac{\varphi(t,z)-\varphi(0,z)}{t\psi(t)}dt
\end{split}
\ee
as $\ve\to 0+$.
  By \eqref{prp} the first and the second summand in  RHS of (\ref{esk3})
  can be estimated by the constant which does not depend on $|z|\in\Z^3$.
  Let us estimate the third summand in  RHS of (\ref{esk3}).
  For $|z|<1/\de$ this summand also can be estimated by the constant.
  For $|z|>1/\de$ we obtain by  \eqref{prp} that
  \begin{align*}
\int\limits_{-\delta}^{\delta}\Big|\frac{\varphi(t,z)-\varphi(0,z)}{t\psi(t)}\Big|dt
&=\int\limits_{|t|<1/|z|}\dots \, +\int\limits_{1/|z|<|t|<\delta}\dots
\\
&\le\frac 1{|z|}C|z|+C\ln|z|\le C\ln|z|.
\end{align*}
Lemma is proved.
  \end{proof}
  Lemma \ref{SP} implies that
  $P_{j}(\om\pm i\ve,z)\to P_{j}(\om\pm i0,z)$ as $\varepsilon\to 0+$, $\forall z\in\Z^3$
  and
  $$
   \sup_{\ve\in (0,1]}|P_{j}(\om\pm i\ve,z)|\le C_j\Big(\ln(1+ |z|)+1\Big).
  $$
  Evidently, that the whole resolvent $R_{0}$ satisfies the similar properties.
  Hence by the  Lebesgue dominated convergence theorem
  \begin{equation*}
      \sum\limits_{x,y\in\Z^3}(1+|x|^2)^{-\sigma}
      |R_0(\om\pm i\ve,x-y)-R_0(\om\pm i0,x-y)|^2
      (1+|y|^2)^{-\sigma}\to 0,\;\varepsilon\to 0+
   \end{equation*}
  with  $\sigma>3/2$.
  Then  the Hilbert-Schmidt norm of the difference $R_0(\om\pm i\ve)-R_0(\om\pm i0)$
   converges to zero. Proposition  \ref{LAP} in the case $k=0$ is proved.\\
  ii) In the case  $k\not=0$ we use integration by parts. For instance, let us consider
  $k=1$.
  Since $\nabla\phi(\theta)\not=0$ for $\theta\in D_j$ then there exists $i\in\{1,2,3\}$
  such that $\pa_i\phi(\theta)\not=0$ for $\theta\in D_j$. Hence,
  \begin{align*}
P'_{j}(\omega\pm i\ve,z)&=\frac{1}{8\pi^3}\int\limits_{D_j}
     \frac{\chi_j(\theta)e^{-i\theta z}}{(\phi(\theta)-\omega\mp i\ve)^2}~d\theta
\\
&=-\fr{1}{8\pi^3}\int\limits_{D_j}\pa_i\Big(\frac{1}{\phi(\theta)-\omega\mp i\ve}\Big)
     \fr{\chi_j(\theta)e^{-i\theta z}~d\theta}{\pa_i\phi(\theta)}
\\
&=\fr{1}{8\pi^3}\int\limits_{D_j}\frac{1}{\phi(\theta)-\omega\mp i\ve}
     \pa_i\Big(\fr{\chi_j(\theta)e^{-i\theta z}}{\pa_i\phi(\theta)}\Big)~d\theta.
\end{align*}
The further proof is similar to the case $k=0$. Differentiating the exponent implies
additional factor $z_i$ and then the value of $\si$ increase on one unit.
  \end{proof}
\subsection{Asymptotics near critical points}
\label{as}
Further we need the information on behavior of the resolvent $R_0(\om)$
near the critical points $\om_k$. We consider ``elliptic'' points
$\om_1=0$, $\om_4=12$ and ``hyperbolic'' points $\om_2=4$, $\om_3=8$ separately.
\subsubsection{Elliptic points}
\label{as-ell}
Here we construct the Puiseux expansion of the free resolvent $R_0(\omega)$
near the point $\om_1=0$
(the expansion near the point $\om_4=12$ can be construct similarly).
\begin{prop}\label{R0-expell}
  Let $N=0,1,2...$ and  $\sigma>N+3/2$. Then the following expansion holds
  in ${\cal B}(\sigma,-\sigma)$:
  \be\label{expR0}
     R_0(\omega)= \sum\limits_{k=0}^{N}A_{k}\omega^{k/2}
     +{\cal O}(\om^{(N+1)/2}),\;|\om|\to 0,\; \arg\om\in(0,2\pi).
  \end{equation}
Here $A_{k}\in{\cal B}(\sigma,-\sigma)$ with $\sigma>k+1/2$.
\end{prop}
\begin{proof}
  The resolvent $R_0(\omega\pm i0)$ is represented by the integral (\ref{R0}).
  Fix  $0<\delta<1$ and consider $0<|\omega|<\delta^2/2$.
  We identify $T^3$ with the cube $[-\pi,\pi]^3$ and represent
  $R_0(\omega,z)$, $z=x-y$, as the sum
 $$
    R_0(\omega,z)=\frac 1{8\pi^3}\int\limits_{B_{\delta}}
    \frac{e^{-i\theta z}}{\phi(\theta)-\om}~d\theta+\frac 1{8\pi^3}\int\limits_{T^3\setminus B_{\delta}}
    \frac{e^{-i\theta z}}{\phi(\theta)-\om}~d\theta=R_{01}(\om, z)+R_{02}(\om, z),
$$
  where $B_{\delta}$ is the ball of radius $\delta$.
Since  $\phi(\theta)=|\theta|^2+{\cal O}(|\theta|^4)$, then
 $R_{02}(\omega,z)$ is analytic function of $\omega$ in $|\omega|\le\delta^2/2$, and
 $$
    |\partial_\om ^j R_{02}(\om,z)|\leq\frac{C_{j,N}}{(|z|+1)},
    \quad |\om|\le\delta^2/2,\quad z\in   \Z^3.
 $$
 Hence it suffices to prove the asymptotics of type (\re{expR0}) for $R_{01}$.
 For simplicity we suppose that $\phi(\theta)=|\theta|^2$ (In the case
 $\phi(\theta)=|\theta|^2+{\cal O}(|\theta|^4)$ the scheme of proving is similar and
 differ only the technical details).
 Let us choose the system of coordinate in which the direction of the axe
 $\theta_3$ coincides with the direction of vector  $z$ and rewrite $R_{01}(\om,z)$ as
 $$
    R_{01}(\om,z)= \frac 1{8\pi^3}\int\limits_{|n|=1}\Big(\int\limits_0^{\de}
    \frac{e^{-ir|z|n_3}r^2dr}{r^2-\om}\Big)dS(n).
$$
 Here $r=|\theta|$, $\theta=rn$.  Then
 \be\la{R2}
\begin{split}
R_{01}(\om,z)&=\frac 1{8\pi^3}\!\int\limits_{S_{+}}\bigg(\int\limits_0^{\de}
    \frac{ e^{-ir|z|n_3}r^2dr}{r^2-\om}
    +\int\limits_0^{\de}
    \frac{e^{ir|z|n_3}r^2dr}{r^2-\om}\bigg)dS(n)
\\
&=\frac 1{8\pi^3}\int\limits_{S_{+}}\bigg(\!\int\limits_{-\de}^{\de}
    \frac{e^{-ir|z|n_3}r^2dr}{r^2 -\om}\bigg)dS(n),
\end{split}
\ee
where $S_{+}=\{|n|=1,\;n_3>0\}$.
The integrand in the RHS of (\re{R2}) has one simple pole
  at the lower half-plane. Let us apply the Cauchy residue theorem:
  \be\la{R3}
\begin{split}
R_{01}(\om,z)&=-\frac {i\sqrt\om}{32\pi^4}\int\limits_{S_1^{+}}
    e^{i\sqrt{\om}|z|n_3}dS(n)
    +\frac 1{8\pi^3}\int\limits_{S_{+}}\bigg(\int\limits_{\Gamma_{\de}}
    \frac{e^{-ir|z|n_3}r^2~dr}{r^2-\om}\bigg)dS(n)
\\
&=R_{01}^1(\om,z)+R_{01}^2(\om,z).
\end{split}
\ee
Here $\Gamma_{\de}=\{|r|=\de,\Im r<0\}$. For the first summand the
  asymptotics of type (\re{expR0}) are evident.
  Let us consider the second summand in the RHS of (\re{R3}):
  \beqn\nonumber
    R_{01}^{2}(\om,z)&=&\frac 1{8\pi^3}\int\limits_{0}^{2\pi}d\beta
    \int\limits_{0}^{\pi/2}d\al
    \int\limits_{\Gamma_{\de}}
    \frac{e^{-ir|z|\cos\al}r^2\sin\al~ dr}{r^2-\om}\\
    \nonumber
    &=&\frac i{4\pi^2|z|}\int\limits_{\Gamma_{\de}}\frac{rdr}{r^2-\om}
    \int\limits_{0}^{\pi/2}~de^{-ir|z|\cos\al}
    =\frac i{4\pi^2|z|}\int\limits_{\Gamma_{\de}}\fr{r(1-e^{-ir|z|})dr}{r^2-\om},
    \quad |z|\not=0,
  \eeqn
  where $\alpha$ is the angle between $z$ and $\theta$.
  Since $r^2=\delta^2$ on $\Gamma_{\de}$, then the operator value function
  $R_{01}^{2}(\om)$ is analytic in $|\om|<\de^2/2$. Moreover, the function $R_{01}^{2}(\om)$
  and all its derivatives in respect to
  $\om$ are bounded in ${\cal B}(\sigma,-\sigma)$ with  $\sigma>1/2$.
  Hence $R_{01}^{2}$ admits an expansion of  type (\re{expR0}).
\end{proof}
\begin{remark}\label{dif1}
  The expansion  \eqref{expR0} can been differentiated $N+1$ times in
  ${\cal B}(\sigma,-\sigma)$ with $\sigma>N+3/2$:
  $$
    \partial_\om^r R_0(\om)=\partial_\om^r\Big(\sum\limits_{k=0}^{N}A_k\om^{k/2}
    \Big)+{\cal O}(\om^{(N+1)/2-r}),\quad 1\le r\le N+1.
  $$
\end{remark}
\begin{proof}
For the proof let us note, that each differentiation of the resolvent in respect to $\om$
increase the power of pole of the integrand in the RHS of (\re{R2}) on one unit. Therefore,
the calculation of the corresponding residue leads to differentiation of the exponent and
then to appear extra factors $|z|$. Hence the value of $\si$ increase  on one unit.
\end{proof}
\subsubsection{Hyperbolic points}
\label{as-hyp}
Here we construct the Puiseux expansion of the free resolvent $R_0(\omega)$ near the
``hyperbolic" point $\om_2=4$ (the expansion near the  point $\om_3=8$ can be construct similarly).
The main contribution into (\ref{R0}) is given by the corresponding critical points
$(0,0,\pi)$, $(0,\pi,0)$ and $(\pi,0,0)$ of hyperbolic type.
\begin{prop}\label{R0-exphyp}
   Let $N=-1,0,1,...$ and $\sigma>2N+7/2$. Then in  ${\cal B}(\sigma,-\sigma)$
    the expansion holds:
   \begin{equation}\label{expR0h}
     R_0(4+\om)=
     \!\!\sum\limits_{k=0}^{N}E_{k}\om^k
     +\sqrt\om\sum\limits_{k=0}^{N}B_{k}\om^k
     +{\cal O}(\om^{N+1}),\;|\om|\to 0,\;\Im\om>0.
  \end{equation}
  Here the operators $E_{k}, B_{k}\in{\cal B}(\sigma,-\sigma)$ with $\sigma>2k+3/2$.
  In the case $\Im\om<0$ the similar expansion  holds.
 \end{prop}
\begin{proof}
  For $\om=\om_2=4$ the denominator of the integral (\ref{R0}) vanishes along the curve
  $\phi(\theta)=4$.  We will study main contribution of points
  $(0,0,\pi)$, $(0,\pi,0)$ and $(\pi,0,0)$ of the curve which are critical points of
  $\phi(\theta)$. The contribution of other points of the curve  can be proved by
  methods of Section \ref{lim-ab-pr}. For concreteness, let us consider the integral
  over a neighborhood of the point $(\pi,0,0)$.
  Let $\zeta(\theta)$ be a smooth cutoff function, equal 1 in a neighborhood of the point
  $(\pi,0,0)$ (the other properties of $\zeta(\theta)$ we specified below). For
  $\Im\om> 0$ denote
  \begin{align*}
    Q(\om,z)&=\frac 1{8\pi^3}\int\limits\frac{e^{-iz\theta}~\zeta(\theta)
    ~d\theta}{\phi(\theta)-4-\om}
\\
&=\frac{e^{-iz_3\pi}}{8\pi^3}\int
    \frac{e^{-iz\theta'}\zeta_1(\theta')~d\theta'}
    {4\sin^2\frac{\theta_1}2+4\sin^2\frac{\theta_2}2-4\sin^2\frac{\theta_3'}2-\om},
\end{align*}
where $\theta'_3=\theta_3-\pi$, $\theta'=(\theta_1,\theta_2,\theta_3')$, è
  $\zeta_1(\theta')=\zeta(\theta)$.
  We suppose that $\zeta_1(\theta')$ is symmetric in $\theta_1$,
  $\theta_2$ and $\theta'_3$.
  Then  the exponent in the numerator
  can be  substituted by its even part, so we have
  \begin{align*}
    Q(\om,z)&=\frac{e^{-iz_3\pi}}{\pi^3}\int\limits_{0}^{\infty}
    \int\limits_{0}^{\infty}\int\limits_{0}^{\infty}
    \frac{\cos(z_{1}\theta_{1})~\cos(z_{2}\theta_{2})~\cos(z_{3}\theta_{3})
    ~\zeta_1(\theta)d\theta}
    {4\sin^2\frac{\theta_1}2+4\sin^2\frac{\theta_2}2-4\sin^2\frac{\theta_3}2-\om}
\\
&=\frac{e^{-iz_3\pi}}{\pi^3}Q_1(\om,z).
\end{align*}
Let us obtain the expansion of type \eqref{expR0h} for $Q_1$.
  We change  the variables: $s_i=2\sin \fr{\theta_i}2$, and choose the cutoff function
  $\zeta$ such that $\zeta_1(\theta)=\zeta_2(|s|^2)$,
  with smooth function $\zeta_2$. Then
  \begin{align*}
Q_1(\om,z)&=\int\limits_{0}^{\infty}\int\limits_{0}^{\infty}\int\limits_{0}^{\infty}
    \frac {F(z,s_1^2,s_2^2,s_3^2)\zeta_2(|s|^2)~ds}
    {s_1^2+s_2^2-s_3^2-\om},
\\
F(z,s_1^2,s_2^2,s_3^2)&=\prod\limits_i \fr{2\cos(2z_i\arcsin
s_i/2)}{\sqrt{4-s_i^2}}.
\end{align*}
Now we use  the cylindrical variables:
  $s_1=\tau\cos\varphi, s_2=\tau\sin\varphi, s_3=s_3$. Then
\be\la{Q1}
\begin{split}
Q_1(\om,z)&=\int\limits_{0}^{\infty}\!\int\limits_{0}^{\infty}\!\frac
    {F_1(z,\tau^2,s_3^2)\zeta_2(\tau^2\!+\!s_3^2)\tau d\tau ds_3}
    {\tau^2-s_3^2-\om},
\\
F_1(z,\tau^2,s_3^2)&=\int\limits_0^{\pi/2}
     F(z,\tau^2\cos^2\!\varphi,\tau^2\sin^2\varphi,s_3^2)d\varphi.
\end{split}
\ee
We change  the variables once more:
  $$
    \rho_1=\tau^2-s_3^2=R^2\cos 2\psi,\quad\rho_2=2\tau s_3=R^2\sin 2\psi,
  $$
  where $R$, $\psi$ are the polar coordinates on the plane $(\tau,s_3)$.
  Then $|\rho|^2=\rho_1^2+\rho_2^2=R^4$, hence,
  $|\rho|=R^2,\;~\tau^2=(|\rho|+\rho_1)/2,\;~ s_3^2=(|\rho|-\rho_1)/2,\;~
  d\rho_1 d\rho_2=4|\rho|d\tau ds_3$ and
  \be\la{h-int}
    Q_1(\om,z)=\int\limits_{0}^{\infty}\Big(\int\limits_{\R}
    \frac{h(|\rho|,\rho_1,z)}{(\rho_1-\om)|\rho|}\sqrt{|\rho|+\rho_1}
    d\rho_1\Big)d\rho_2,
  \ee
  where
  $h(|\rho|,\rho_1,z)=F_1(z,\fr{|\rho|+\rho_1}2,
  \fr{|\rho|-\rho_1}2)\zeta_2(|\rho|)/4\sqrt 2$.
  Now we can specify all needed properties of cutoff function:
  $$
    \supp\zeta_2(|\rho|)\cap\{\rho\in\R^2: \rho_2\ge 0\}
    \subset\Pi=\{(\rho_1,\rho_2): -\delta\le\rho_1\le\delta,\,0\le\rho_2\le\delta\}
  $$
  with some $0<\delta<1$. We consider $0<|\om|\le\de/2$, $\Im\om>0$.
  Denote $r=|\rho|$. The function $h(r,\rho_1,z)$
  can be expanded into the following finite Taylor series with respect to $\rho_1$:
  \be\la{hk}\
    h(r,\rho_1,z)=h_0(r,z)+h_1(r,z)\rho_1+...
    +h_N(r,z)\rho_1^N+H_{N}(r,\rho_1,z)\rho_1^{N},
  \ee
  where $h_k(r,z)$ are polynomial in $z$ of order $2k$, and
  \be\la{HN-est}
\begin{split}
|H_{N}(r,\rho_1,z)|&\le C|z|^{2N},
\\
|\pa_{\rho_1} H_{N}(r,\rho_1,z)|&\le C|z|^{2N+2}, \quad
(\rho_1,\rho_2)\in [-\de,\de]\times [0,\de].
\end{split}
\ee
 Let us substitute  (\ref{hk}) into (\ref{h-int}). Then
 \be\la{JJ}
  Q_1(\om,z)=\sum\limits_{k=0}^N J_k(\om,z)+{\tilde J_N}(\om,z),
\ee
where
\begin{align*}
  J_k(\om,z)&=\int\limits_{\Pi}
    \frac{h_k(r,z)\rho_1^k\sqrt{r+\rho_1}}{(\rho_1-\om)r}d\rho_1 d\rho_2,
\\
{\tilde J_N}(\om,z)&=\int\limits_{\Pi}\!
    \frac{H_{N}(r,\rho_1,z)\rho_1^{N}d\rho_1 d\rho_2}
    {(\rho_1-\om)r}\sqrt{r+\rho_1}.
\end{align*}
{\bf Step i)}.
  First we consider the summands  $J_k(\om,z)$, $k=0,1,...,N$:
  \beqn\nonumber
    J_k(\om,z)
    \!\!\!&=&\!\!\!\int\limits_{\Pi}\!\!
    \frac{h_k(r,z)\sqrt{r+\rho_1}}{r}
    \Big(\rho_1^{k-1}\!+\!\om \rho_1^{k-2}+\!...
    +\!\om^{k-1}\!\!+\!\frac{\om^k}{\rho_1-\om}\Big) d\rho_1 d\rho_2\\
    \la{hk-int}
    \!\!\!&=&\!\!\!\sum\limits_{k=0}^{k-1} a_{k,j}(z)\om^j
    +\om^k\int\limits_{\Pi}\frac{h_k(r,z)\sqrt{r+\rho_1}}
    {(\rho_1-\om)r}d\rho_1 d\rho_2\\
    \nonumber
    \!\!\!&=&\!\!\!\sum\limits_{j=0}^{k-1} a_{k,j}(z)\om^j+
    \om^k\int\limits_0^\de\frac{ h_k(r,z)}{\sqrt{r}}dr\int\limits_{0}^{\pi}
    \frac{\sqrt{1+\cos\psi}~d\psi}{\cos\psi-\om/r},
  \eeqn
where $a_{k,j}(z)$ are polynomial  of order $2k$.
Let us calculate the integral
\be\la{tan}
\begin{split}
\int\limits_{0}^{\pi}\frac{\sqrt{1+\cos\psi}~d\psi}{\cos\psi-\om/r}
    &=\int\limits_0^{\pi}\!\frac{2\sqrt 2d(\sin\frac{\psi}2)}
    {1-2\sin^2\frac{\psi}2-\frac{\om}{r}}
\\
&=\int\limits_0^1\!\frac{-\sqrt 2dt}{t^2-\frac{r-\om}{2r}}
    =\frac{-\sqrt{r}}{\sqrt{r-\om}}
    \log\frac{1\!-\!\sqrt{\frac{r-\om}{2r}}}{1\!+\!\sqrt{\frac{r-\om}{2r}}}
    +\frac{\pi i\sqrt{r}}{\sqrt{r-\om}}.
\end{split}
\ee
Here $\sqrt r\ge0$, function  $z=\sqrt{r-\om}$ is  analytic in
  $\Im\om>0$ with the values in $\Im z<0,\;\Re z>0$, and function
  $\zeta=\log w$ is analytic in $|w|<1,\;\Im w>0$, where
  $\log(-1)=\pi i$.

  Substitute (\ref{tan}) into (\ref{hk-int}), we get
  \be\la{hk-int1}
    J_k(\om,z)
    =\sum\limits_{j=0}^{k-1} a_{k,j}(z)\om^j+\om^k\!\int\limits_0^{\de}
    \left(\pi i-\log\frac{1-\sqrt{\frac{r-\om}{2r}}}
    {1+\sqrt{\frac{r-\om}{2r}}}\right)\frac{h_k(r,z)dr}{\sqrt{r-\om}}
  \ee
Let us expand $h_k(r,z)$ into the following finite Taylor series with respect to  $r$:
  \be\la{hkj}
    h_k(r,z)=h_{k,0}(z)+h_{k,1}(z)r+...+h_{k,N-k}(z)r^{N-k}+
    H_{k,N-k}(r,z)r^{N-k},
  \ee
  where $h_{k,j}(z)$ are polynomial of order $2(k+j)$,
  and $|H_{k,N-k}(r,z)|\le C|z|^{2N}$, $0\le r\le\delta$.
  The following lemma is true
  \begin{lemma}\la{IK}
    Let $0<|\om|<\de/2$, $\Im\om>0$. Then
    \be\la{Ik}
      I_l=\int_0^{\de}\left(\pi i-\log\frac{1-\sqrt{\frac{r-\om}{2r}}}
      {1+\sqrt{\frac{r-\om}{2r}}}\right)
      \frac{r^ldr}{\sqrt{r-\om}}=s_{l}(\om)+C_l\om^l\sqrt{\om},
    \ee
    where $s_{l}$  are analytic in $0<|\om|<\de/2$,
    $\Im\om>0$, $C_l\in\R$.
  \end{lemma}
  We shall prove this lemma in Appendix A.
  Now \eqref{hk-int1}-\eqref{Ik} imply that for
  $0<|\om|<\de/2$, $\Im\om>0$
  \beqn\nonumber
  J_k(\om,z)&=&\sum\limits_{j=0}^{N} b_{k,j}(z)\om^j
    +\om^k\sqrt\om\sum\limits_{j=0}^{N-k}c_{k,j}(z)\om^j
    +\tilde a_{N,k}(\om,z)\om^{N+1}\\
    \la{hk-int2}
    &+&\om^k\int\limits_0^{\de}
    \left(\pi i\!-\!\log\frac{1-\sqrt{\frac{r-\om}{2r}}}
    {1+\sqrt{\frac{r-\om}{2r}}}\right)
    \frac{H_{k,N-k}(r,z)r^{N-k}dr}{\sqrt{r-\om}},
  \eeqn
  where $|b_{k,j}(z)|\le C|z|^{2N}$,  $|c_{k,j}(z)|\le C|z|^{2(k+j)}$,
  and $|\tilde a_{N,k}(\om,z)|\le C|z|^{2N}$.
Further,
  \be\la{HNk}
    \int\limits_0^{\de}
    \left(\pi i-\log\frac{1-\sqrt{\frac{r-\om}{2r}}}
    {1+\sqrt{\frac{r-\om}{2r}}}\right)
    \frac{H_{k,N-k}(r,z)r^{N-k}dr}{\sqrt{r-\om}}
    =\int\limits_0^{2|\om|}+\int\limits_{2|\om|}^{\de}={\cal I}_1+{\cal I}_2.
  \ee
  In ${\cal I}_1$ we change the variable: $r=|\om|\tau$. Then
  \be\la{I1}
\begin{split}
|{\cal I}_1|&=|\om|^{N-k}\sqrt{|\om|}~\left|\int\limits_0^2
    \left(\!\pi i-\log\frac{1-\sqrt{\frac{\tau-\om/|\om|}{2\tau}}}
    {1+\sqrt{\frac{\tau-\om/|\om|}{2\tau}}}\right)
    \frac{H_{k,N-k}(|\om|\tau,z)\tau^{N-k}d\tau}
    {\sqrt{\tau-\om/|\om|}}\right|
\\
&\le C|z|^{2N}|\om|^{N-k}\!\sqrt{|\om|}.
\end{split}
\ee
Let us expend $\sqrt{r-\om}$ and  the function in brackets  into the finite Taylor series
with respect to $\om/r$. Then
\be\la{I2}
\begin{split}
{\cal I}_2&=\int\limits_{2|\om|}^{\de}H_{k,N-k}(r,z)r^{N-k-1/2}
\\
&\quad\times\Big(d_0+d_1\frac{\om}{r}+\dots
+d_{N-k}\frac{\om^{N-k}}{r^{N-k}}
+\hat d_{N-k}(\om/r)\frac{\om^{N-k}}{r^{N-k}}\Big)dr
\\
&=\int\limits_{2|\omega |}^{\delta }H_{k,N-k}(r,z)
\\
&\quad\times\Big(d_0r^{N-k-1/2}+d_{1}\omega
    r^{N-k-3/2}+\dots
+d_{N-k}\omega ^{N-k}r^{-1/2}\Big)dr+\tilde u_{N-k}(\om,z)
\\
&=\int\limits_{0}^{\delta }\!H_{k,N-k}(r,z)
\\
&\quad\times\Big(d_0r^{N-k-1/2}+d_{1}\omega
    r^{N-k-3/2}+\dots
+d_{N-k}\omega ^{N-k}r^{-1/2}\Big)dr+\hat u_{N-k}(\om,z)
\\
&=\sum\limits_{j=0}^{N-k}u_j(z)\om^j+\hat u_{N-k}(\om,z),
\end{split}
\ee
where
  $|\hat d_{N-k}(\om/r)|\le C,\;~|u_j(z)|\le C|z|^{2N},\;~|\tilde u_{N-k}(\om,z)|,
  \;|\hat u_{N-k}(\om,z)|\le C|z|^{2N}|\om|^{N-k}$.
  Now (\ref{hk-int2})-(\ref{I2}) imply that
  \be\la{Jk}
  J_k(\om,z)=\sum\limits_{j=0}^{N} d_{k,j}(z)\om^j
    +\om^k\sqrt\om\sum\limits_{j=0}^{N-k}c_{k,j}(z)\om^j
    +\tilde d_{N,k}(\om,z),
  \ee
where $|d_{k,j}(z)|\le C|z|^{2N}$, $|c_{k,j}(z)|\le C|z|^{2(k+j)}$
and $|\tilde d_{N,k}(\om,z)|\le C|z|^{2N}|\om|^N$.\\
{\bf Step ii)}.
  It remains to consider the summand ${\tilde J_N}(\om,z)$ in the RHS of (\ref{JJ}):
  \be\la{JN}
\begin{split}
{\tilde J_N}(\om,z)&=\int\limits_{\Pi}
    \frac{H_{N}(r,\rho_1,z)\sqrt{r+\rho_1}}{r}
\\
&\quad\times\Big(\rho_1^{N-1}+\om\rho_1^{N-2}
    +\dots+\om^{N-1}+\frac{\om^N}{\rho_1-\om}\Big)d\rho_1 d\rho_2
\\
&=\sum\limits_{j=0}^{N-1}w_j(z)\om^j
+\om^{N}\int\limits_{\Pi}
\frac{H_{N}(r,\rho_1,z)\sqrt{r+\rho_1}~d\rho_1 d\rho_2}{(\rho_1-\om)r},
\end{split}
\ee
where $|w_j(z)|\le C|z|^{2N}$.
The following estimate is true
\begin{lemma}\la{rem-est1}
  Let $0<|\om|<\de/2$, $\Im\om>0$. Then
  \be\la{Rem-est1}
    |\int\limits_{\Pi}
    \frac{H_{N}(r,\rho_1,z)\sqrt{r+\rho_1}~d\rho_1 d\rho_2}{(\rho_1-\om)r}|
    \le C|z|^{2N}\ln^2|z|,\; |z|>1.
  \ee
\end{lemma}
We shall prove the lemma in Appendix B.\\
{\bf Step iii)}.
  Finally, (\ref{JJ}), (\ref{Jk})-(\ref{Rem-est1}) imply that
  $$
    Q_1(\omega,z)=\sum\limits_{k=0}^{N}
    q_k(z)\omega^k+\sqrt\om\sum\limits_{k=0}^{N}p_k(z)\omega^k
     +\widehat q_{N}(\omega,z),\;|\om|\to 0,
  $$
  where $|\widehat q_{N}(\omega,z)|\le C|z|^{2N}\ln^2|z||\om|^{N}$.
  Further, $p_k(z)={\cal O}(|z|^{2k})$, and $q_k(z)={\cal O}(|z|^{2N})$
  for $0\le k\le N$. Therefore, $q_k(z)={\cal O}(|z|^{2k})$, since $q_k(z)$
  do not depend on $N$.
\end{proof}
\begin{cor}\la{32}
Let $\si>3/2$. Then in ${\cal B}(\sigma,-\sigma)$ the expansion holds:
  \begin{equation}\label{exp-h}
     R_0(4+\om)={\cal O}(1),\;|\om|\to 0,\;\Im\om>0.
  \end{equation}
\end{cor}
\begin{cor}\la{diff2}
The expansion  \eqref{expR0h} can be differentiated. More precisely,
\be\la{pa1}
\partial_\om R_0(4+\om)=\fr{B_0}{2\sqrt\om}+{\cal O}(1),\quad|\om|\to
0,\quad\Im\om>0,
\ee
in ${\cal B}(\sigma,-\sigma)$ with $\sigma>7/2$,
\be\la{pa2}
    \partial_\om^2 R_0(4+\om)=-\fr{B_0}{4\om\sqrt\om}
    +{\cal O}(\om^{-1/2}),\quad|\om|\to 0,\quad\Im\om>0,
\ee
in ${\cal B}(\sigma,-\sigma)$ with $\sigma>11/2$.
\end{cor}
\begin{proof} It is sufficient to obtain the asymptotics of type (\ref{pa1}) and  (\ref{pa2})
for $Q_1$ defined in (\ref{Q1}).
Formula  (\ref{Q1}) implies
\be\la{fint}
\begin{split}
&\pa_\om Q_1(\om,z)=\int\limits_{0}^{\infty}\!\int\limits_{0}^{\infty}\!\frac
    {F_1(z,\tau^2,s^2)\zeta_2(\tau^2+s^2)\tau d\tau ds}
    {(\tau^2-s^2-\om)^2}
\\
&\quad=\int\limits_{0}^{\infty}\!ds\int\limits_{0}^{\infty}
    F_1(z,\tau^2,s^2)\zeta_2(\tau^2+s^2)\pa_\tau\frac{-1/2}{\tau^2\!-\!s^2-\om}d\tau
\\
&\quad=-\fr 12\int\limits_{0}^{\infty}\fr{F_1(z,0,s^2)\zeta_2(s^2)}{s^2+\om}~ds
+\int\limits_{0}^{\infty}\!\int\limits_{0}^{\infty}
\fr{\pa_\tau^2\big(F_1(z,\tau^2,s^2)\zeta_2(\tau^2+s^2)\Big)\tau}
{\tau^2-s^2-\om}~d\tau ds
\\
&\quad=S_1(\om,z)+S_2(\om,z).
\end{split}
\ee
The asymptotics of  $S_2(\om,z)$ are similar to (\ref{exp-h}):
\be\la{S2}
S_2(\om)={\cal O}(1),\quad|\om|\to 0,\quad\Im\om>0,
\ee
in ${\cal B}(\sigma,-\sigma)$ with $\sigma>7/2$.
Let us note, that the differentiation in respect to $\tau^2$
implies extra factors $|z|^2$ and then the value of $\si$
increase  on two units by comparison with (\ref{exp-h}).

Consider $S_1(\om,z)$:
\be\la{S1}
S_1(\om,z)=\int\limits_{0}^{\de}\fr{(F_2(z,s^2)-F_2(z,|\om|))ds}{s^2+\om}
+\int\limits_{0}^{\de}\fr{F_2(z,|\om|)ds}{s^2+\om}
\ee
For the function $F_2(z,s^2)=-F_1(z,0,s^2)\zeta_2(s^2)/2$ the bounds hold:
\be\la{Fz}
|F_2(z,s^2)|\le C,~~~|\pa_{s^2}F_2(z,s^2)|\le C|z|^{2},~~~ |\pa_{s^2}^2 F_2(z,s^2)|\le C|z|^{4}.
\ee
Let us estimate the first integral in the RHS of (\ref{S1}) using the second bound (\ref{Fz}):
$$
\Big|\int\limits_{0}^{\de}\fr{(F_2(z,s^2)-F_2(z,|\om|))ds}{s^2+\om}\Big|
\le C|z|^{2}\int\limits_{0}^{\de}\fr{|s^2-|\om||}{|s^2+\om|}ds\le C|z|^{2}
$$
Let us calculate the second integral in the RHS of (\ref{S1}):
\be\la{pq}
\int\limits_{0}^{\de}\fr{F_2(z,|\om|)ds}{s^2+\om}=F_2(z,|\om|)\fr 1{2\sqrt{-\om}}
\big(\ln\fr{\de-\sqrt{-\om}}{\de+\sqrt{-\om}}-\pi i\big)
=p(z)\fr 1{\sqrt\om}+q(\omega,z),
\ee
where
$$
|p(z)|+|q(\omega,z)|\le C,\quad 0<|\om|<\de/2,~~\Im\om>0.
$$
Therefore,
\be\la{S11}
S_1(\om)=\fr {P_1}{\sqrt\om}+{\cal O}(1),\;|\om|\to 0,\;\Im\om>0
\ee
in ${\cal B}(\sigma,-\sigma)$ with $\sigma>7/2$.
From (\ref{S2}) è (\ref{S11}) the asymptotics for the first derivative follow.
Further, let us consider the second derivative:
\be\la{fin}
\begin{split}
&\pa_\om^2 Q_1(\om,z)=\int\limits_{0}^{\infty}\!\int\limits_{0}^{\infty}\!\frac
{2F_1(z,\tau^2,s^2)\zeta_2(\tau^2+s^2)\tau d\tau ds}
{(\tau^2-s^2-\om)^3}
\\
&\quad=\int\limits_{0}^{\infty}\!ds\!\int\limits_{0}^{\infty}\!\!
F_1(z,\tau^2,s^2)\zeta_2(\tau^2\!+s^2)\pa_\tau\frac {-1/2}{(\tau^2-s^2-\om)^2}d\tau
\\
&\quad=\fr 12\int\limits_{0}^{\infty}\fr{F_1(z,0,s^2)\zeta_2(s^2)}{(s^2+\om)^2}~ds
+\int\limits_{0}^{\infty}\int\limits_{0}^{\infty}
\fr{\pa_{\tau^2}\big(F_1(z,\tau^2,s^2)\zeta_2(\tau^2+s^2)\Big)\tau}
{(\tau^2-s^2-\om)^2}~ d\tau ds.
\\
&\quad=U_1(\om,z)+U_2(\om,z).
\end{split}
\ee
The proof of  asymptotics for $U_2(\om,z)$ is similar to the proof of asymptotics for the first
derivative. We obtain
$$
U_2(\om)=\fr {P_2}{\sqrt\om}+{\cal O}(1),\quad|\om|\to 0,\quad\Im\om>0,
$$
in ${\cal B}(\sigma,-\sigma)$ with $\sigma>11 /2$.
Let us consider $U_1(\om,z)$:
\begin{align*}
&U_1(\om,z)=-\int\limits_{0}^{\de}\fr{F_2(z,s^2)}{(s^2+\om)^2}~ds
\\
&\quad=-\int\limits_{0}^{\de}\fr{F_2(z,s^2)-F_2(z,|\om|)-F'_2(z,|\om|)(s^2-|\om|)}{(s^2+\om)^2}~ds
+\int\limits_{0}^{\de}\fr{F_2(z,|\om|))}{(s^2+\om)^2}~ds
\\
&\quad+\int\limits_{0}^{\de}\fr{F_2'(z,|\om|))}{s^2+\om}~ds
=U_{11}(\om,z)+U_{12}(\om,z)+U_{13}(\om,z).
\end{align*}
The last bound  (\ref{Fz}) implies
\be\la{U1}
U_{11}(\om,z)|\le C|z|^4.
\ee
Further,  we obtain  similar to (\ref{pq})
\be\la{U2}
\begin{split}
U_{13}(\om,z)&=\fr {p_1(z)}{\sqrt\om}+q_1(\om,z),\quad
|p_1(z)|+|q_1(\omega,z)|
\\
&\le C|z|^2,\quad 0<|\om|<\fr{\de^2}2,\enskip \Im\om>0.
\end{split}
\ee
Finally,
\be\la{U3}
\begin{split}
&U_{12}(\om,z)=F_2(z,|\om|)\Big(\fr 1{4\om\sqrt{-\om}}
\big(\pi i-\ln\fr{1+\sqrt{-\om}/\de}{1-\sqrt{-\om}/\de}\big)+\fr{\de}{2\om(\de^2+\om)}\Big)
\\
&=F_2(z,|\om|)\Big(\fr {1}{4\om\sqrt{-\om}}\big(\pi i
-\fr{2\sqrt{-\om}}{\de}+\fr{2\om\sqrt{-\om}}{3\de^3}+\dots\big)
+\fr 1{2\om\de}\Big(1-\fr{\om}{\de^2}+\dots\Big)\Big)
\\
&=\fr {s_2(z)}{\om\sqrt\om}+\fr {p_2(z)}{\sqrt\om}+q_2(\om,z),\quad
|s_2(z)|+|p_2(z)|+|q_2(\omega,z)|\le C,
\\
&\quad 0<|\om|<\fr{\de^2}2,\quad\Im\om>0.
\end{split}
\ee
From (\ref{U1})--(\ref{U3}) the asymptotics for
the second derivative follow.
\end{proof}
\setcounter{equation}{0}
\section{Perturbed resolvent }
\label{lap-pert}
\subsection{Limiting absorption principle}
\label{l-ab-pr}
For the perturbed resolvent the limiting absorption principle holds.
\begin{prop}\label{BV}
  Let $V\in {\cal V}$, $\sigma>3/2$.
  Then  the  limits exist  as $\varepsilon\to 0+$
  \begin{equation}\label{esk}
     R(\omega\pm i\varepsilon)\;\buildrel {\hspace{2mm}B(\sigma,-\sigma)}\over
     {- \hspace{-2mm} \longrightarrow}\; R(\omega\pm i0),\quad
     \omega\in \Si\setminus \{0,4,8,12\}.
  \end{equation}
\end{prop}
\begin{proof}
Let $\omega\in\Si\setminus \{0,4,8,12\}$ and $\sigma>3/2$.
Since the potential $V$ has a finite support, then
Proposition \ref{LAP} implies
\begin{equation*}
    I+VR_0(\omega\pm i\varepsilon)\;\buildrel {\hspace{2mm}\mathcal
    B(\sigma,\sigma)}\over
    {- \hspace{-2mm} \longrightarrow}\;I+VR_0(\omega\pm i0),\quad
    \varepsilon\to 0+.
  \end{equation*}
The operator $I+VR_0(\omega\pm i0)$ has only a trivial kernel (see \cite[Theorem 10]{ShV}),
and the operator $VR_0(\omega\pm i0)$ is finite dimensional. Hence, the operator
$I+VR_0(\omega\pm i0)$ is invertible, and moreover
\begin{equation}\la{sh}
   \bigl(I+VR_0(\omega\pm i\varepsilon)\bigr)^{-1}
   \;\buildrel {\hspace{2mm}\mathcal B(\sigma,\sigma)}\over
   {- \hspace{-2mm} \longrightarrow}\;\bigl(I+VR_0(\omega\pm i0)
   \bigr)^{-1},\quad\varepsilon\to 0+.
\ee
Further, (\ref{lim}), (\ref{sh}) and identity  $R=R_0(I+VR_0)^{-1}$ imply the existence
of the limits \eqref{esk}.
\end{proof}
\begin{remark}\label{diff3}
 For $\omega\in \Si\setminus \{0,4,8,12\}$ the derivatives
 $\partial_\om^kR(\omega\pm i0)$ belong to ${\cal B}(\sigma,-\sigma)$ with $\sigma>3/2+k$.
\end{remark}
\begin{proof}
The statement follows from Proposition \ref{LAP} and the identity
(see \cite[Theorem 9.2]{jeka})
$$
R^{(k)}=\Big[(1-RV)R_0^{(k)}-\left(\begin{array}{c} k \\ 1\end{array}\right)
R'VR_0^{k-1}-...\Big](1-VR),\quad k=1,2,...
$$
\end{proof}
\subsection{Asymptotics near critical points}
\label{puispert-sect}
In this  sections we are going to obtain an asymptotic expansion
for the perturbed resolvent $R(\omega)$ near the critical points $\om_k$, $k=1,2,3,4$.
\begin{definition}\label{gener-def}
i) A set $\cal W\subset\cal V$ is called generic, if for each $V\in\cal V$
we have $\alpha V\in\cal W$, with the possible exception of a discrete set of $\alpha\in\R$.
\\
ii) We say that a property holds for a ``generic'' $V$,
if it holds for all $V$ from a generic subset  $\cal W\subset\cal V$.
\end{definition}
\begin{theorem}\label{resas}
  Let $\sigma >3/2$. Then for ``generic''  $V\in {\cal V}$ the following
  expansion holds:
  \begin{equation}\label{Asymp1}
    R(\om)=D_1+{\cal O}(\sqrt\om),\;~~~|\om|\to 0,\;\arg\om\in (0,2\pi)
  \end{equation}
  in the norm of ${\cal B}(\sigma,\,-\sigma)$.
\end{theorem}
\begin{proof}
  We use the relation
  \be\la{T-def}
    R(\om)=T^{-1}(\om)R_{0}(\om),\quad \mbox {\rm where} \quad T(\om):=I+R_{0}(\om)V.
  \ee
  According to  \eqref{expR0},
  \be\la{T-rep}
    T(\om)=I+A_{0}V +O(\sqrt\om),\quad |\om|\to 0,\quad\arg\om\in (0,2\pi).
  \end{equation}
  Let us prove that for ``generic''  $V\in {\cal V}$ the operator $T(\om)$ is invertible
  in $l^2_{-\si}$ for sufficient small $|\om|>0$.
  It is suffices to prove that  the operator
  $ T(0)=I+A_{0}V$ is invertible in $l^2_{-\si}$,
  or the operator with the kernel
  $$
    (1+x^2)^{-\sigma/2}(\delta(x-y)+A_{0}V(y))(1+y^2)^{\sigma/2}
  $$
  is invertible in $l^2$. Let us consider the operator
  $$
    {\cal A}(\alpha)={\rm Op}[(1+x^2)^{-\sigma/2}
    \bigl(\delta(x-y)+\alpha A_{0}V(y)\bigr)(1+y^2)^{\sigma/2}]
    =1+\alpha {\cal K},~~~~~~~\;\alpha\in\C.
  $$
  For $\sigma>3/2$
  $$
    K(x,y)=(1+x^2)^{-\sigma/2}A_{0}V(y)(1+y^2)^{\sigma/2}\in l^2(\Z^2\times\Z^2).
  $$
  Hence,  $K(x,y)$ is a Hilbert-Schmidt kernel, and accordingly the operator
  ${\cal K}\!=\!{\rm Op}(K(x,y))$: $l^2\to l^2$ is compact.
  Further,
  ${\cal A}(\alpha)$ is analytic in $\alpha\in\C$, and ${\cal A}(0)$ is invertible.
  It follows that ${\cal A}(\alpha)$ is invertible for all
  $\alpha\in\C$ outside a discrete set; see \cite{Bl}.
  Thus we could replace the original potential $V$ by
  $\alpha V$ with $\alpha$ arbitrarily close to $1$,
  if necessary, to have $T(0)$ invertible.

  Now (\ref{T-def}) and (\ref{T-rep}) imply that for sufficiently small  $|\om|>0$
  \be\label{rw}
    R(\om) =(I+T(0)+{\cal O}(\sqrt\om))^{-1}(A_{0}+{\cal O}(\sqrt\om))
    =T(0)^{-1}A_{0}+O(\sqrt\om).
  \ee
\end{proof}
\begin{remark}\la{rmr}
i) The expansion of the resolvent near the second elliptic point
   $\om_4= 12$ is similar to the expansion (\ref{Asymp1}).\\
ii)  The expansion of type (\ref{Asymp1}) near the hyperbolic points $\om_2=4$ and $\om_3=8$
   require larger value of $\sigma$. Namely, for ``generic''  $V\in {\cal V}$
$$
R(\om_k+\om)=D_k+{\cal O}(\sqrt\om),\quad |\om|\to 0,\quad\Im\om>0,\quad k=1,2
$$
in the norm of ${\cal B}(\si,\,-\si)$ with $\si>7/2$.\\
iii) These expansion can be differentiated two times in ${\cal B}(\si,\,-\si)$ with
$\si>5/2$ for elliptic points and with $\si>11/2$ for hyperbolic points.
In these cases $\pa^2_\om R(\om_k+\om)={\cal O}(\om^{-3/2})$, $k=1,2,3,4$.
\end{remark}
\setcounter{equation}{0}
\section{Long-time asymptotics}
\label{lt-as}
Now we  apply Lemma \ref{jkl} below, which is a version of Lemma 10.2 from \cite{jeka}
to prove the following theorem
\begin {theorem}\label{end}
Let $\sigma>11/2$. Then  for ``generic''  $V\in {\cal V}$
the following asymptotics  hold
  \be\la{full2}
    \norm{e^{-itH}-\sum\limits_{j=1}^n e^{-it\mu_j} P_j}
    {B(\sigma,-\sigma)}={\cal O}(t^{-3/2}),\quad t\to\infty.
  \ee
Here  $P_j$ denote the projections on the eigenspaces
corresponding to the eigenvalues $\mu_j\in\R\setminus [0,12],\;j=1,\dots,n$.
\end{theorem}
\begin{proof}
The estimate \eqref{full2} is based on the formula
  \begin{equation}\label{Decay}
    e^{-itH}=-\displaystyle\frac 1{2\pi i}\oint\limits_{|\omega|=C} e^{-it\omega}
    R(\om)d\omega,\;C>\max\{12;|\mu_j|,\,j=1,...,n\}.
  \end{equation}
  The integral above is equal to the sum of residues at the poles of $R(\om)$
and the integral over the contour around the segment [0,12], i.e.
\begin{align*}
e^{-itH}-\sum\limits_{j=1}^n e^{-it\mu_j} P_j
    &= \displaystyle\frac 1{2\pi i}\int\limits_{[0, 12]}
    e^{-it\omega}( R(\omega+i0)- R(\omega-i0))\,d\omega
\\
&=\int\limits_{[0,12]} e^{-it\omega} P(\omega)d\omega.
\end{align*}
The main contribution into the long-time asympotics gives the integrals
 over the neighbourhoods of  the critical points.
 For example, let us consider the integral over the neighbourhood of  the  point
 $\om_1=0$. Expansion  \eqref{Asymp1} and  Remark \ref{rmr} imply
  \begin{equation}\label{P-asy}
    \pa^k P(\omega)={\cal O}(\pa^k\sqrt\omega),\;\omega\to +0,\;\omega\in\R,
   \; k=0,1,2
  \end{equation}
 in  ${B(\sigma,-\sigma)}$ with $\si>11/2$.
 The following result is a special case of \cite[Lemma 10.2]{jeka}.
\begin{lemma}\label{jkl} Assume ${\cal B}$ is a Banach space, $a>0$,
and $F\in C(0, a; {\cal B})$ satisfies $F(0)=F(a)=0$, $F'\in L^1(0,a; {\cal B})$,
as well as $F''(\omega)={\cal O}(\omega^{-3/2})$ as $\omega\to 0$.
Then
\[ \int\limits_{[0,a]} e^{-it\omega}F(\omega)d\omega ={\cal O}(t^{-3/2}),
   \quad t\to\infty. \]
\end{lemma}
\bigskip
Set $F(\om)=\zeta(\omega)P(\omega)$, where $\zeta$ is the smooth function,
$\zeta(\om)=1$ for $\om\in [-1/2,1/2]$,
$\supp\zeta\in(-1,1)$; $a=1$, ${\cal B}=B(\si,-\si)$ with $\si>11/2$.
Then due to (\ref{P-asy})  we can apply Lemma \ref{jkl} to get
$$
\int\limits_{[0,1]} e^{-it\om}\zeta(\om)P(\om)={\cal O}(t^{-3/2}),\quad t\to\infty,
$$
in  ${B(\sigma,-\sigma)}$ with $\si>11/2$.

The integrals over  the neighborhoods of  the other critical points
can be estimated similarly.
\end{proof}

\setcounter{equation}{0}
\section{Klein-Gordon equation}
\label{KG-sect}
Now we extend the results of Sections \ref{lap-pert}-\ref{lt-as}
to the case of the Klein-Gordon equation \eqref{KGE}.
Denote ${\bf\Psi}(t)\equiv \bigl(\psi(\cdot,t),\dot\psi(\cdot,t)\bigr)$,
${\bf\Psi}_0 \equiv \bigl(\psi_0,\pi_0\bigr)$. Then \eqref{KGE} becomes
$$
i\dot{\bf\Psi}(t)={\bf H \Psi}(t)=\left(
  \begin{array}{cc}
      0                 &   i\\
      i(\Delta-m^2-V)   &   0
  \end{array}\right){\bf\Psi}(t),\quad t\in\R;\quad{\bf\Psi}(0)={\bf\Psi}_0,
$$
The resolvent ${\bf R}(\om)$ can be expressed in term of the resolvent
$R(\om)$ as
\begin{equation}\label{KGE4}
    {\bf R}(\omega)=
    \left(
    \begin{array}{cc}
            \omega R(\omega^2-m^2)           &   iR(\omega^2-m^2)
    \\
         -i(1 +\omega^2 R(\omega^2-m^2))      &   \omega R(\omega^2-m^2)
     \end{array}
     \right),\quad \omega^2-m^2\in\C\setminus [0,12].
   \end{equation}
\noindent
Representation  \eqref{KGE4} and the properties of $R(\om)$ imply
the following long time asymptotics:

Let $\sigma>11/2$ and ${\bf\Psi}_0\in l^2_\si\oplus l^2_\si$.
Then for ``generic''  $V\in {\cal V}$
\begin{equation*}
   \norm{\,e^{-it{\bf H}}{\bf\Psi}_0-\sum\limits_{\pm}\sum\limits_{j=1}^n
   e^{-it\nu_j^{\pm}}{\bf P}_j^{\pm}{\bf\Psi}_0}
   {l^2_{-\si}\oplus l^2_{-\si}}={\cal O}(t^{-3/2}),\quad t\to\infty.
\end{equation*}
Here  ${\bf P}_j^{\pm}$ are the projections onto the eigenspaces
corresponding to the eigenvalues
$\nu_j^{\pm}=\pm\sqrt{m^2+\mu_j}$, $j=1,\ldots, n$.
\section{Appendix A}
Let us prove Lemma \ref{IK} by induction. For $l=0$ we get
\begin{align*}
I_0&=\int\limits_0^{\de}
 \left(\pi i-\log\frac{1-\sqrt{\frac{r-\om}{2r}}}{1+\sqrt{\frac{r-\om}{2r}}}\right)
  \frac{dr}{\sqrt{r-\om}}
\\
&=2\sqrt{r-\om}
  \left(\pi i-\log\frac{1-\sqrt{\frac{r-\om}{2r}}}{1+\sqrt{\frac{r-\om}{2r}}}\right)
  \Big|_0^{\de}+2\sqrt{2}\om\int\limits_0^{\de}\frac{dr}{\sqrt{r}(r+\om)}
\\
&=2\sqrt{\de-\om}\left(\pi i-
  \log\frac{1-\sqrt{\frac{\de-\om}{2\de}}}{1+\sqrt{\frac{\de-\om}{2\de}}}\right)
  -i2\sqrt{2\om}\log\frac{\sqrt{r}-i\sqrt{\om}}{\sqrt{r}+i\sqrt{\om}}\Big|_0^{\de}
\\
&=\ti s_0(\om)
  -i2\sqrt{2\om}\log\frac{1-i\sqrt{\frac{\om}{\de}}}{1+i\sqrt{\frac{\om}{\de}}}
  -\pi\sqrt{2\om}=s_{0}(\om)+C_0\sqrt{\om},
\end{align*}
where $\ti s_{0}$, $s_{0}$ are the analytic functions of  $\om$, $C_0=-\pi\sqrt 2$.

Further, for $l\ge 1$ we get
\begin{align*}
I_l&=\int_0^{\de}\!
  \left(\pi i-\log\frac{1-\sqrt{\frac{r-\om}{2r}}}{1+\sqrt{\frac{r-\om}{2r}}}\right)
  \frac{r^ldr}{\sqrt{r-\om}}=2\de^l\sqrt{\de-\om}
  \left(\pi i-\log\frac{1-\sqrt{\frac{\de-\om}{2\de}}}
  {1+\sqrt{\frac{\de-\om}{2\de}}}\right)
\\
&-2l\!\int\limits_0^{\de}\!\frac{r^{l-1}(r-\om)}{\sqrt{r-\om}}
  \left(\pi i-\log\frac{1-\sqrt{\frac{r-\om}{2r}}}{1+\sqrt{\frac{r-\om}{2r}}}\right)dr
  +2\sqrt{2}\om\int\limits_0^{\de}\frac{r^{l}dr}{\sqrt{r}(r+\om)}
\\
&=\ti s_l(\om)-2lI_l+2l\om I_{l-1}
\\
&+2\sqrt{2}\om\int\limits_0^{\de}\frac{dr}{\sqrt r}
  \Big(r^{l-1}-\om r^{l-2}+\dots+(-\om)^{l-1}+\frac{(-\om)^{l}}{r+\om}\Big)
\\
&=\ti s_l(\om)-2lI_l+2l\om I_{l-1}+\ti{\ti s}_l(\om)
\\
&-i2\sqrt{2\om}(-\om)^{l}
  \log\frac{1-i\sqrt{\frac{\om}{\de}}}{1+i\sqrt{\frac{\om}{\de}}}
  -2\pi\sqrt{2\om}(-\om)^{l},
\end{align*}
where $\ti s_l$, $\ti{\ti s}_l$  are the analytic functions of $\om$.
Hence,  $I_l=s_{l}(\om)+C_l\om^l\sqrt{\om}$,
where $s_{l}$ are the analytic functions of $\om$ and $C_l\in\R$.
\section{Appendix B}
Here we prove Lemma \ref{rem-est1}.
We estimate only the integral over $\Pi_+=\{0\le\rho_1,\rho_2\le\de\}$.
The integral over $\Pi\setminus\Pi_+$ can be estimated similarly.
Let us split the integral over $\Pi_+$ into two integrals:
\begin{align*}
&\int\limits_{\Pi_+}\frac{H_{N}(r,\rho_1,z)\sqrt{r+\rho_1}
  ~d\rho_1 d\rho_2}{(\rho_1-\om)r}
\\
&\qquad=\int\limits_{\Pi_+} \frac{(H_{N}(r,\rho_1,z)-H_{N}(r,|\om|,z))\sqrt{r+\rho_1}
  d\rho_1 d\rho_2}{(\rho_1-\om)r}
\\
&\qquad+\int\limits_{\Pi_+}
  \frac{H_{N}(r,|\om|,z)\sqrt{r+\rho_1}d\rho_1 d\rho_2}{(\rho_1-\om)r}
  =J_1+J_2.
\end{align*}
Similar to (\ref{tan}) we obtain
$$
  J_2=\int\limits_0^{\de}
  \left(\pi i-\log\frac{1-\sqrt{\frac{r-\om}{2r}}}{1+\sqrt{\frac{r-\om}{2r}}}\right)
  \frac{H_{N}(r,|\om|,z)dr}{\sqrt{2(r-\om)}}.
$$
Note that
$$
  \Big|\log\Big|\frac{\sqrt{2r}-\sqrt{r-\om}}{\sqrt{2r}+\sqrt{r-\om}}\Big|\Big|
  =\Big|\log\Big|\frac{r+\om}{(\sqrt{2r}+\sqrt{r-\om})^2}\Big|\Big|
$$
$$
  \le |\log|r+\om||+2|\log|\sqrt{2r}+\sqrt{r-\om}||
  \le 2|\log|r-|\om|||.
$$
Then \eqref{prp} implies
$$
  |J_2|\le C|z|^{2N}\int\limits_0^{\de}\frac{1+|\log|r-|\om|||}{|\sqrt{r-|\om|}|}dr
  \le C|z|^{2N}.
$$
Further, for $|z|>1$ let us split $J_1$ as
$$
  J_1=J_{11}+J_{12}+J_{13},
$$
where  $J_{11}$ is integral over
$$
\Pi_1=\{(\rho_1,\rho_2)\in\Pi_+: |r|<1/|z|^{4/3}\},
$$
$J_{12}$ is integral over
$$
\Pi_2=\{(\rho_1,\rho_2)\in\Pi_+\setminus\Pi_1: |\rho_1-|\om||<1/|z|^{8/3}\},
$$
and $J_{13}$ is integral over
$$
\Pi_3=\Pi_+\setminus(\Pi_1\cup\Pi_2)
$$
(see Picture 1).
\begin{figure}  
\vspace{-4cm}   
\begin{center}
\hspace*{1.8cm}
\includegraphics[width=15cm]{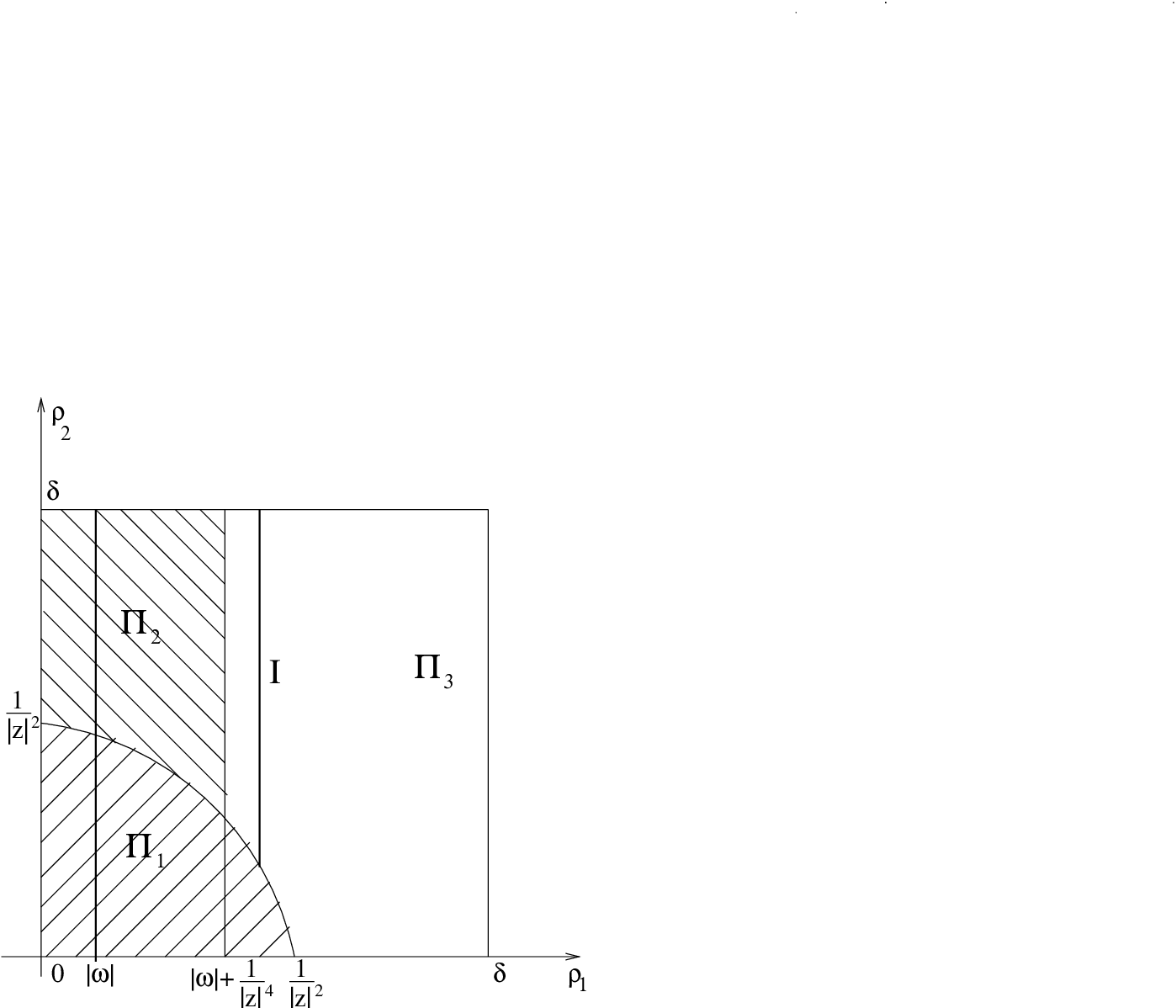}  
\end{center}
\caption{Case $|\om|-1/|z|^{8/3}<0$.}
\label{Picture1}
\end{figure}

By (\ref{HN-est}) and the inequality $|\rho_1-\om|\ge|\rho_1-|\om||$
we get
\begin{align*}
  |J_{11}|&\le C |z|^{2N+2}
  \int\limits_{\Pi_1}\frac{|\rho_1-|\om||\sqrt{r+\rho_1}d\rho_1 d\rho_2}
  {|\rho_1-\om|r}
  \le C |z|^{2N+2}
  \int\limits_{\Pi_1}\frac{\sqrt{r+\rho_1}d\rho_1 d\rho_2}{r}
\\
&\le C |z|^{2N+2}\int\limits_0^{\pi/2}d\psi\int\limits_0^{1/|z|^{4/3}}\!\!\!
  \sqrt{r+r\cos\psi}~dr
\\
&\le C |z|^{2N+2}\int\limits_0^{\pi/2}\cos\frac\psi 2d\psi
  \int\limits_0^{1/|z|^{4/3}}\sqrt rdr
  \le C |z|^{2N}.
\end{align*}
$$
  \le C |z|^{2N+2}\int\limits_0^{\pi/2}d\psi\int\limits_0^{1/|z|^{4/3}}\!\!\!
  \sqrt{r+r\cos\psi}~dr
  \le C |z|^{2N+2}\int\limits_0^{\pi/2}\cos\frac\psi 2d\psi
  \int\limits_0^{1/|z|^{4/3}}\sqrt rdr
  \le C |z|^{2N}.
$$
For the second integral we obtain similarly
\begin{align*}
  |J_{12}|&\le C |z|^{2N+2}\!\!\int\limits_{\Pi_2}\!\!\frac{\sqrt{r
  +\rho_1}d\rho_1 d\rho_2}{r}
\\
&\le C |z|^{2N+2}\!\!\int\limits_{\Pi_2}\frac{d\rho_1 d\rho_2}{\sqrt r}
  \le C |z|^{2N+2}|z|^{2/3}\frac{\de}{|z|^{8/3}}\le C |z|^{2N},
\end{align*}
since $1/\sqrt r\le |z|^{2/3}$ for $(\rho_1,\rho_2)\in\Pi_2$,
and $|\Pi_{2}|\le 2\de/|z|^{8/3}$. Finally, (\ref{HN-est}) implies
$$
  |J_{13}|\le C |z|^{2N}\int\limits_{\Pi_3}\frac{d\rho_1 d\rho_2}
  {|\rho_1-|\om||\sqrt{\rho_1^2+\rho_2^2}}\le C |z|^{2N}\ln^2|z|,
$$
since for any vertical interval $I\in\Pi_3$ we get
$$
  \int\limits_I\frac{d\rho_2}{\sqrt{\rho_1^2+\rho_2^2}}=\ln(\rho_2+\sqrt{\rho_1^2
  +\rho_2^2})\le C\ln|z|.
$$
\setcounter{equation}{0}
\section{Appendix C. Asymptotic completeness}
\label{Scat-sect}
We apply the obtained results to construct  the asymptotic scattering states.
Let $u_k$ be the eigenfunctions of operator $H$, corresponding  eigenvalues
$\mu_k$, and $U_0(t)$ be the dynamical group of free Schr\"odinger equation.

\begin{theorem}
i) Let $\sigma>11/2$ and $\psi_0\in l^2_{\sigma}$. Then  for ``generic''  $V\in {\cal V}$
   for solution to \eqref{Schr} the following long time asymptotics hold
   \be\la{scat}
    \psi(\cdot,t)=\sum\limits_{k=1}^n C_k e^{-it\mu_k}u_k+U_0(t)\phi_\pm+r_\pm(t),
    \quad t\to\pm\infty,
  \ee
where  $\phi_\pm\in l^2$ are the corresponding scattering states, and
$$
\Vert r_\pm(t)\Vert_{l^2}={\cal O}(|t|^{-1/2})
$$
\end{theorem}
\begin{proof}
For concreteness we  consider the case $t\to +\infty$.
Let us apply the projector $P^c$
onto the continuous spectrum of the operator $H$ to both sides of \eqref{Schr}:
\be\la{proj}
    iP^c\dot\psi=P^c H\psi=H_0 P^c\psi+V P^c\psi
\ee
 since $P^c$ and $H$ commute.
 Applying  the Duhamel representation to equation \eqref{proj} we obtain
  \be\la{Dug}
    P^c\psi(t)= U_0(t)P^c\psi(0)+\int\limits_0^t U_0(t-\tau)VP^c\psi(\tau)d\tau, ~~~~t\in\R.
  \ee
  We can rewrite  \eqref{Dug} as
  \begin{align*}
    P^c\psi(t)&= U_0(t)\Big(P^c\psi(0)+\int\limits_0^{\infty}
    U_0(-\tau)VP^c\psi(\tau)d\tau\Big)
\\
&-\int\limits_t^{\infty} U_0(t-\tau)VP^c\psi(\tau)d\tau
    =U_0(t)\phi+r_+(t).
\end{align*}
Let us show  that the  integrals  converge,
  and the function
  $$
  \phi_+=P^c\psi(0)+\int\limits_0^{\infty} U_0(-\tau)VP^c\psi(\tau)d\tau
  $$
  belongs to $l^2$. Indeed,  \eqref{full2} implies
  \beqn\nonumber
    \int\limits_0^{\infty}\Vert U_0(-\tau)VP^c\psi(\tau)\Vert_{l^2}d\tau
    &=& \int\limits_0^{\infty}\Vert VP^c\psi(\tau)\Vert_{l^2}d\tau
    \le C\int\limits_0^{\infty}\Vert P^c\psi(\tau)\Vert_{l^2_{-\sigma}}d\tau\\
    \nonumber
    &\le& C\int\limits_0^{\infty}t^{-3/2}\Vert \psi(0)\Vert_{l^2_{\sigma}}d\tau
    \le C
  \eeqn
Here we used  the unitarity of $U_0(t)$ in  $l^2$ and the identity
  $P^c=I-P^d$, where $P^d$ is the projector onto the discrete  spectrum, which consists
  of the exponentially decreasing functions.
  The estimate for $r_+(t)$ follows similarly.
\end{proof}
For the Klein-Gordon equation the asymptotics of type (\ref{scat}) also hold.


\begin{thebibliography}{14}
\bibitem{Bl}
{\sc Bleher P.M.}, On operators depending meromorphically on a
parameter, {\em Moscow Univ.~Math.~Bull.~}{\bf 24}, 21-26 (1972)

\bibitem{E}
{\sc Eskina M.S.}, The scattering problem for partial-difference equations,
in {\em Mathematical Physics} Naukova Dumka, Kiev, 1967 (in Russian), 248-273

\bibitem{Gl}
 {\sc Glazman I.M.:} Direct methods of qualitative spectral analysis of
 singular differential operators, Israel Program for scientific translation,
 Jerusalem 1965 and by Daniel Davey and Co., New York 1966.

\bibitem{IV}
{\sc Islami H.,\,\ Vainberg B.},
Large time behavior of solutions to difference wave operators,
{\em J.~ Commun.~ Partial Differ.~ Equations}{\bf\,\ 31}, no.1-3, 397-416 (2006)

\bibitem{jeka}
{\sc Jensen A.,\,\ Kato T.}, Spectral properties of Schr\"odinger operators
and time-decay of the wave functions, {\em Duke Math.~J.}{\bf\,\,46},
583-611 (1979)

\bibitem{jene}
{\sc Jensen A.,\,\ Nenciu G.}, A unified approach to resolvent expansions
at thresholds,{\em Reviews in Math.~Physics}{\bf\,\,13}, no.6, 717-754 (2001)

\bibitem{kkk}
{\sc Komech A.,\,\ Kopylova E.,\,\ Kunze M.}, Dispersive estimates for 1D discrete
Schr\"odinger and Klein-Gordon equations, {\em J.~Appl.~Anal.}{\bf\,\,85}, no.12,
1487-1508 (2006)

\bibitem{kkv}
{\sc Komech A.,\,\ Kopylova E.,\,\ Vainberg B.}, Dispersive estimates for 2D discrete
Schr\"odinger and Klein-Gordon equations, {\em J.~Funct.~Anal.}{\bf\,\,254}, no.8,
2227-2254 (2008)
\bibitem{LP}
{\sc Lax P.,\,\ Phillips R.}, Scattering Theory. Academic Press,
New York (1989)

\bibitem{M}
{\sc Murata M.}, Asymptotic expansions in time for solutions of
Schr\"o\-din\-ger-type equations, {\em J.~Funct.~Anal.}{\bf\,\,49},
10-56 (1982)

\bibitem{schlag}
{\sc Schlag W.}, Dispersive estimates for Schr\"odinger operators:
A survey, preprint {\tt math.AP/0501037}

\bibitem{ShV}
{\sc Shaban W.,\,\ Vainberg B.}, Radiation conditions for the difference
Schr\"odinger operators, {\em J.~Appl.~Anal.}{\bf\,\,80}, no.3-4,
525-556 (2001)

\bibitem{V}
{\sc Vainberg B.:} On the short wave asymptotic behaviour of solutions of steady-state
 problems and the asymptotic behaviour as $t\to\infty$ of solutions of non-stationary
 problems, {\em Russian Math.Surveys} {\bf\,\,30}, no.2, 1-58 (1975)
\bibitem{V2}

{\sc Vainberg B.:} Asymptotic Methods in Equations of Mathematical Physics.
Gordon and Breach, New York (1989)

\end{thebibliography}
\end{document}